\documentclass[9pt,shortpaper,twoside,web]{ieeecolor}
\usepackage{generic}

\usepackage{amsmath} 
\usepackage{amssymb} 
\newcommand {\matr}[2]{\left[\begin{array}{@{}#1@{}}#2\end{array}\right]}

\usepackage{lipsum}
\usepackage{mathtools}
\usepackage{cuted}
\usepackage{algorithm}
\usepackage[noend]{algpseudocode}
\usepackage{graphicx}
\graphicspath{{images/}}

\usepackage{tabularx,booktabs}
\usepackage{accents}
\usepackage{expl3,xparse}

\newtheorem{Theorem}{Theorem}
\newtheorem{Lemma}[Theorem]{Lemma}

\newtheorem{Corollary}[Theorem]{Corollary}
\newtheorem{Assumption}[Theorem]{Assumption}
\newtheorem{Definition}[Theorem]{Definition}
\newtheorem{Remark}{\it{Remark}}

\usepackage{tikz}
\usepackage{cite}
\usepackage{soul}

\begin{document}

\title{A Gauss-Newton-Like Hessian Approximation for Economic NMPC}

\author{Mario Zanon
	\thanks{
		M. Zanon is with IMT School for Advanced Studies Lucca, Italy. e-mail: mario.zanon@imtlucca.it. The author thanks Rien Quirynen, S\'ebastien Gros, Robin Verschueren, Moritz Diehl for the fruitful discussions on Economic NMPC, and nonlinear programming; and Robert Hult for the fruitful discussions on energy-optimal autonomous driving.
	}%
}

\maketitle


\begin{abstract}
	Economic Model Predictive Control (EMPC) has recently become popular because of its ability to control constrained nonlinear systems while explicitly optimizing a prescribed performance criterion. Large performance gains have been reported for many applications and closed-loop stability has been recently investigated. However, computational performance still remains an open issue and only few contributions have proposed real-time algorithms tailored to EMPC. We perform a step towards computationally cheap algorithms for EMPC by proposing a new positive-definite Hessian approximation which does not hinder fast convergence and is suitable for being used within the real-time iteration (RTI) scheme. We provide two simulation examples to demonstrate the effectiveness of RTI-based EMPC relying on the proposed Hessian approximation.
\end{abstract}

\begin{IEEEkeywords}
	Economic model predictive control.
\end{IEEEkeywords}

\section{Introduction}\label{sec:intro}

Model Predictive Control (MPC) is an advanced control technique which can explicitly account for nonlinear constrained dynamics while minimizing a prescribed cost. Traditionally the problem has been formulated as a \emph{tracking} formulation minimizing some distance from a prescribed reference. 
Instead, in so-called \emph{economic} formulations the cost is not directly related to a setpoint, but rather to a performance index that should be optimized. While an improvement in closed-loop performance is expected over tracking formulations, Economic MPC (EMPC) poses challenges both in terms of stability guarantees and computational burden: the former has been widely studied, see, e.g.,~\cite{Diehl2011,Amrit2011a,Muller2013a,Mueller2015,Zanon2013d,Zanon2017e,Grune2013a, Stieler2014b, Faulwasser2015g, Faulwasser2015a,Zanon2018a,Faulwasser2018,Gruene2018}; however, the latter has been only partially investigated, see, e.g.,~\cite{Quirynen2016,Quirynen2017a,Shin2019,Verschueren2017}.

The main algorithmic challenges in EMPC stem from the impossibility of relying on a Gauss-Newton Hessian approximation and the need to compute second-order sensitivities and the regularize the Lagrangian Hessian to ensure positive-definiteness.
The first issue has been investigated in, e.g.,~\cite{Quirynen2016}, while the second one has been investigated in ~\cite{Verschueren2017}. A tracking scheme called Locally Equivalent To Economic MPC (LETEMPC)~\cite{Zanon2014d,Zanon2016b,Zanon2017a}, delivers a first-order approximation of the economically optimal feedback law and, therefore, yields only approximate economic optimality.

In this paper, we present a new Hessian approximation for EMPC which does not require any additional online computation with respect to tracking MPC and avoids the computation of second-order sensitivities and Hessian regularizations. The proposed Hessian approximation, however, is such that the algorithm retains good convergence properties and can be seen as a sort of Gauss-Newton Hessian approximation for economic MPC, even though the cost is not of least-squares type.
The proposed algorithm solves the economic MPC problem to full optimality, and can also be seen as the approximate economic MPC approach~\cite{Zanon2016b,Zanon2017a} with a gradient correction which guarantees full economic optimality. 

This paper is structured as follows. We introduce the problem in Section~\ref{sec:formulation} and in Section~\ref{sec:preliminaries} we establish a set of preliminary results. We introduce the new Hessian approximation in Section~\ref{sec:hessian_approximation}, where we also prove stability of RTI-based EMPC, provided that our Hessian approximation is used. We demonstrate the theoretical results in simulations in Section~\ref{sec:simulations} and conclude in Section~\ref{sec:conclusions}.

\section{Problem Formulation and Main Contribution}
\label{sec:formulation}
We consider nonlinear discrete-time systems 
\begin{align*}
	x_{k+1}=f(x_k,u_k), && x\in\mathbb{R}^{n_x}, \ u\in\mathbb{R}^{n_u},
\end{align*}
that shall be operated such that constraints $h(x_k,u_k)\geq0$ are satisfied and the cost $\sum_{k=0}^\infty \ell(x_k,u_k)$ is minimized. 
MPC approximates the infinite-horizon problem by optimizing over a finite horizon $N$. At every sampling instant, the state measurement $\hat x_0$ is received, an Optimal Control Problem (OCP) is solved, and the first control input is applied to the system. At the next sampling instant the procedure is repeated to close the feedback loop.

Economic MPC consists in repeatedly solving
\begin{subequations}
	\label{eq:mpc}
	\begin{align}
	w^*:= \arg\min_{w} \ \ & \sum_{k=0}^{N-1} \ell(x_k,u_k) + V_\mathrm{f}(x_N) \\
	\mathrm{s.t.} \ \ & x_0 - \hat x_0=0, \label{eq:mpc_ic}\\
	& x_{k+1} - f(x_k,u_k)=0, && k\in\mathbb{I}_0^{N-1}, \label{eq:mpc_dyn} \\
	& h(x_k,u_k) \geq 0, && k\in\mathbb{I}_0^{N-1},  \label{eq:mpc_pc} \\
	& x_N \in \mathbb{X}_\mathrm{f}, \label{eq:mpc_tc}
	\end{align}
\end{subequations}
where we define vector $w := (w_0,w_1,\ldots,w_{N-1},w_N)$, with $w_k:=(x_k,u_k), \ k\in\mathbb{I}_0^{N-1}$ and $w_N:=x_N$; finally, $\mathbb{I}_a^b:=\{a,a+1,\ldots,b\}$. 
The terminal cost $V_\mathrm{f}$ and constraint~\eqref{eq:mpc_tc} are design parameters.  
The MPC feedback policy is $\pi(\hat x_0)=u^*_0$. Throughout this paper we assume that a minimizer of Problem~\eqref{eq:mpc} exists and all functions are sufficiently smooth, i.e., $f,\, h,\, \ell,\, V_\mathrm{f} \in C^2$.

The main difficulties relative to economic MPC are (a) the difficulty in proving stability and (b) the computational burden associated with it.
Both difficulties stem from the fact that $\ell(x,u) \ngeq \alpha(\|x\|)$. In the following, we label a problem as economic if $\nexists \, \alpha\in \mathcal{K}$ s.t. $\ell(x,u) \geq \alpha(\|x\|)$. For more details on stability proofs for economic MPC, we refer to~\cite{Diehl2011,Amrit2011a,Muller2013a,Mueller2015,Zanon2013d,Zanon2017e,Grune2013a, Stieler2014b, Faulwasser2015g, Faulwasser2015a,Zanon2018a,Faulwasser2018,Gruene2018}. 
In this paper we address problem (b). 

Both issue (a) and (b) are milder in case of tracking MPC, since stability is easier to enforce~\cite{Rawlings2009b} and the least-squares cost makes it possible to deploy efficient algorithms to solve the problem in real-time, including the popular Real-Time Iteration (RTI) scheme. 
Tracking MPC has been widely studied in the literature, see, e.g.,~\cite{Rawlings2009b,Grune2011,Mayne2014} and references therein. The main drawback of tracking MPC is that, since it penalizes deviations from the optimal steady-state $(x_\mathrm{s},u_\mathrm{s})$, typically with a quadratic penalty $\|(x_\mathrm{s},u_\mathrm{s})\|_W^2$, it does not account for performance $\ell(x_k,u_k)$ during transients, such that the closed-loop cost can significantly increase.
In order to combine the benefits of tracking and economic MPC, in~\cite{Zanon2016b} a locally equivalent to economic MPC (LETEMPC) formulation with quadratic cost has been proposed which delivers feedback policy $\pi^\mathrm{t}(\hat x_0)$ satisfying
\begin{align*}
	\| \pi^\mathrm{t}(\hat x_0)-\pi(\hat x_0) \| = O(\|\hat x_0 - x_\mathrm{s}\|^2).
\end{align*}

Inspired by the LETEMPC formulation, we propose a hybrid formulation, i.e., an EMPC formulation which relies on the Gauss-Newton Hessian approximation of LETEMPC, calculated as in~\cite{Zanon2016b}, to reduce online computations. In this context, it is important to underline that exact Hessian requires one not only to compute second-order derivatives, but also to make sure that the reduced Hessian is positive definite, both of which can be computationally demanding. Finally, as a further motivation, some QP solvers require that the full Hessian is positive definite.

\subsection{Main Contribution}

The main contribution of this paper can be summarized as:

\vspace{0.5em}
\begin{center}
	\parbox{0.8\linewidth}{ \emph{ We propose a new Hessian approximation for EMPC which (a) is positive-definite, (b) enjoys approximation properties equivalent to those of Gauss-Newton Hessian approximations, and (c) guarantees statbility when used in combination with the RTI framework. }}
\end{center}

We will formalize this statement in Theorem~\ref{thm:red_hess}, and Theorem~\ref{thm:rti_stability}. The most important implication of Theorem~\ref{thm:red_hess} is that quick convergence can be obtained without the need to compute online second-order sensitivities nor Hessian regularizations enforcing positive-definiteness. Additionally, while the RTI scheme has been successfully applied to economic MPC in practice~\cite{Quirynen2014c,Quirynen2017a,Verschueren2016b,Gros2014b}, the standard stability proof from~\cite{Diehl2005b} does not directly apply to EMPC. We close this gap by proving in Theorem~\ref{thm:rti_stability} that the stability guarantees provided in~\cite{Diehl2005b} extend to RTI-based economic MPC, provided that the proposed Hessian approximation is used. With slight abuse of terminology, we will refer to the proposed Hessian approximation as the Gauss-Newton (GN) Hessian approximation for EMPC, since it is a GN Hessian approximation for LETEMPC.

Since periodic operation might outperform steady-state operation, 
the stability analysis for the steady-state case has been extended to the periodic case~\cite{Zanon2013d,Zanon2017e,Kohler2018,Mueller2016} and a periodic variant of the LETEMPC has been proposed in~\cite{Zanon2017a}. The extension of our setting to the periodic case is possible, but omitted for the sake of simplicity. 



\section{Preliminaries}
\label{sec:preliminaries}

In this section we introduce Sequential Quadratic Programming (SQP) and recall existing results on stability of economic MPC, based on the concept of \emph{strict dissipativity} and \emph{cost rotation}. In the last part of the section, we present a novel insight about the cost rotation and its impact on the SQP iterates, which will be useful next. All developments also apply to the interior-point framework.


\subsection{Sequential Quadratic Programming}
Consider an NLP of the form 
\begin{align}
\label{eq:nlp}
\min_{w} \ \ & J(w) & \mathrm{s.t.} \ \ & \hat g(w) = 0, && \hat h(w) \geq 0,
\end{align}
with Lagrangian
$\mathcal{\hat L}(z) = J(w) - \hat \lambda^\top \hat g(w) - \hat \mu^\top \hat h(w)$, primal-dual variable $z=(w,y)$, and dual variable $y=(\lambda,\mu)$.
Starting from an initial guess $z^{(0)}=(w^{(0)},y^{(0)})$, SQP computes the solution to~\eqref{eq:nlp} iteratively by relying on the update
\begin{align*}
z^{(i+1)} = z^{(i)} + t \Delta z^{(i)}, && \Delta z^{(i)} := (w^{\mathrm{QP}_i}, y^{\mathrm{QP}_i}-y^{(i)}),
\end{align*}
with $t$ a step length and $(w^{\mathrm{QP}_i},y^{\mathrm{QP}_i})$ the optimal solution of
\begin{align*}
\min_{w} \ \ & \frac{1}{2} w^\top L^{(i)} w + \nabla J(w^{(i)})^\top w \\
\mathrm{s.t.} \ \ & \nabla \bar g(w^{(i)})^\top w + \bar g(w^{(i)}) = 0, \\
& \nabla \bar h(w^{(i)})^\top w + \bar h(w^{(i)}) \geq 0.
\end{align*}
Here, $L^{(i)}$ is the Lagrangian Hessian, or a suitably selected approximation. Local minima are characterized by the Strong Second-Order Sufficient Conditions (SSOSC), i.e., $Z^\top L^{(i)} Z \succ 0$, with $Z$ the null space of the Jacobian of the strongly active constraints $Y_{\mathbb{A}_\mathrm{s}}$. This requirement must be enforced throughout the iterates in order to guarantee descent. Therefore, we denote $\left [ \nabla^2_{ww} \mathcal{\hat L} \right ]_+$ the modification of the Hessian of the Lagrangian such that $Z^\top \left [ \nabla^2_{ww} \mathcal{\hat L} \right ]_+Z \succ 0$ holds. Throughout this paper we assume that the Linear Independence Constraint Qualification (LICQ) holds, i.e., $Y_{\mathbb{A}_\mathrm{s}}$ is full row rank. For more details on the topic, we refer the interested reader to~\cite{Nocedal2006}. 

%

For the economic MPC Problem~\eqref{eq:mpc}, we define the primal-dual variables $z:=(w,\lambda,\mu,\nu)$ and the Lagrangian 
\begin{align*}
	\mathcal{L}(z) &:=  \sum_{k=0}^{N} \mathcal{L}_k(z)- \lambda_{0}^\top (x_0-\hat x_0),
\end{align*}
where, for $k\in\mathbb{I}_0^{N-1}$, we define
\begin{align*}
	\mathcal{L}_k(z) &:= \ell(w_k) - \lambda_{k+1}^\top (x_{k+1} - f(w_k)) - \mu_k^\top h(w_k), \  k\in\mathbb{I}_0^{N-1}, \\
	\mathcal{L}_N(z) &:= V_\mathrm{f}(x_N) + \nu^\top g_\mathrm{f}(x_N).
\end{align*}
Function $g_\mathrm{f}$ lumps together the (equality or inequality) constraints defined in~\eqref{eq:mpc_tc}. 
For ease of notation, for any function $a$ we denote $a(x_k,u_k)$ as $a(w_k)$; and $\nabla a_k=\nabla a(w_k)$, $\nabla a_k^{(0)}=\nabla a(w_k^{(0)})$ and $\nabla_{x} a_k=\nabla_x a(x_k,u_k)$, $\nabla_{u} a_k=\nabla_u a(x_k,u_k)$.

\subsection{Economic MPC and Rotated Economic MPC}
\label{sec:preliminaries_stability}
We recall next some concepts used to analyze economic MPC.
We define the (single-stage) Steady-State Optimization Problem (SOP)%
\begin{align}%
\label{eq:steady_state}%
\min_{x,u} \ \ \ell(x,u) && \mathrm{s.t.} \ \ x-f(x,u)=0, && h(x,u) \leq 0,
\end{align}%
and assume, without loss of generality, that the origin is the unique optimal solution, i.e., $w_\mathrm{s}:=(x_\mathrm{s},u_\mathrm{s})=(0,0)$; and%
\begin{align*}%
\ell(x_\mathrm{s},u_\mathrm{s}) = 0, && V_\mathrm{f}(x_\mathrm{s}) = 0 .
\end{align*}%
Finally, we define the SOP multipliers $\lambda_\mathrm{s}$, $\mu_\mathrm{s}$ and Lagrangian as $$\mathcal{L}_\mathrm{s}(w_\mathrm{s},\lambda_\mathrm{s},\mu_\mathrm{s}):=\ell(w_\mathrm{s})+\lambda_\mathrm{s}^\top (x_\mathrm{s}-f(w_\mathrm{s})) + \mu_\mathrm{s}^\top h(w_\mathrm{s}),$$ with $\lambda_\mathrm{s}\neq0$ in general~\cite{Zanon2018a,Faulwasser2018}.

\begin{Definition}[Strict dissipativity]
	System $x_{k+1}=f(x_k,u_k)$ is strictly dissipative with respect to the supply rate $\ell$ if there exists a bounded storage function $\Lambda(x)$ with $\Lambda(x_\mathrm{s})=0$, such that the following inequality is satisfied for all $(x_k,u_k)$ on the domain of the MPC problem~\eqref{eq:mpc}:%
	\begin{align}%
	\label{eq:strict_diss}%
	\Lambda(f(x_k,u_k)) - \Lambda(x_k) \leq -\rho(\|x_k\|) + \ell(x_k,u_k),
	\end{align}%
	where $\rho$ is a positive definite function.
\end{Definition}
Note that, if~\eqref{eq:strict_diss} holds, then~\eqref{eq:steady_state} must have a unique solution.
Given a storage function $\Lambda(x)$ with $\Lambda(x_\mathrm{s})=0$ we define the rotated stage and terminal cost as%
\begin{subequations}%
	\label{eq:rotated_cost}%
	\begin{align}%
	\bar \ell(x_k,u_k) &:= \ell(x_k,u_k) + \Lambda(x_k) - \Lambda(f(x_k,u_k)), \\
	\bar V_\mathrm{f}(x) &:= V_\mathrm{f}(x) + \Lambda(x).
	\end{align}%
\end{subequations}%
\begin{Assumption}
	\label{ass:terminal_stabilizing}
	There exist a compact set $\mathbb{X}_\mathrm{f}$ containing $x_\mathrm{s}$ in its interior and a terminal control law $\kappa_\mathrm{f}$ such that%
	\begin{align*}%
		V_\mathrm{f}(f(x,\kappa_\mathrm{f}(x))) &\leq V_\mathrm{f}(x) - \ell(x,\kappa_\mathrm{f}(x)), && &
		h(x,\kappa_\mathrm{f}(x)) &\geq0,
	\end{align*}%
	hold $\forall \ x\in\mathbb{X}_\mathrm{f}$. Note that this entails that $f(x,\kappa_\mathrm{f}(x)) \in \mathbb{X}_\mathrm{f}$.
\end{Assumption}

\begin{Theorem}[Stability~\cite{Amrit2011a}]
	\label{thm:empc_stability}
	Assume that strict dissipativity and Assumption~\ref{ass:terminal_stabilizing} hold; constraints $h$ define a compact set $\mathcal{Z}$; $f$, $h$, $\Lambda$ and $\ell$ are $C^2$ on $\mathcal{Z}$; and $V_\mathrm{f}$ is $C^2$ on $\mathbb{X}_\mathrm{f}$. Then the origin is an asymptotically stable equilibrium for the closed-loop system.
\end{Theorem}



Using the rotated cost~\eqref{eq:rotated_cost} the rotated MPC problem reads%
\begin{subequations}%
	\label{eq:rmpc}%
	\begin{align}%
	\min_{w} \ \ & \sum_{k=0}^{N-1} \bar \ell(x_k,u_k) + \bar V_\mathrm{f}(x_N) \\
	\mathrm{s.t.} \ \ & \eqref{eq:mpc_ic}, \eqref{eq:mpc_dyn}, \eqref{eq:mpc_pc}, \eqref{eq:mpc_tc},
	\end{align}
\end{subequations}
with Lagrangian $\mathcal{\bar L}(w,\lambda,\mu)$ defined analogously to the one of the original MPC problem.
Equivalence of the primal solutions of the original and rotated MPC problems has been established to prove Theorem~\ref{thm:empc_stability}. While solving the rotated problem is clearly appealing, to do so one must first compute a storage function satisfying~\eqref{eq:strict_diss}. Unfortunately, this is known to be very hard in the general case. Since it affects the SQP convergence properties, in this paper we are also interested in the dual solution. Therefore, we prove the following.
\begin{Lemma}
	\label{lem:rotation}
	The rotated MPC problem~\eqref{eq:rmpc} delivers the same primal solution as the original MPC problem~\eqref{eq:mpc}. The dual solution, however, is different in general and satisfies%
	\begin{align}%
	\label{eq:dual_rotation}%
	\bar \lambda_{k} = \lambda_{k} + \nabla \Lambda(x_k), && \bar \mu_{k} = \mu_{k}, && \bar \nu = \nu.
	\end{align}
\end{Lemma}
\begin{proof}
	The proof of the first claim follows along the lines of~\cite{Diehl2011,Amrit2011a} and is recalled here for the sake of completeness. 
	We expand the rotated cost as follows
	\begin{align*}
	\bar J(w)=&\sum_{k=0}^{N-1} \bar \ell(x_k,u_k) + \bar V_\mathrm{f}(x_N) \\
	=&\sum_{k=0}^{N-1} \ell(x_k,u_k) + V_\mathrm{f}(x_N) + \Lambda(x_0)\\
	=& \ J(w) + \Lambda(\hat x_0),
	\end{align*}
	such that 
	the original and rotated cost differ by the constant $\Lambda(\hat x_0)$.
	
	We now prove the second claim, which has not received much attention so far. For $k\in\mathbb{I}_0^{N-1}$ optimality implies
	\begin{align*}
	0=\nabla_{x_k} \mathcal{L} &= \nabla_{x_k} \ell_k - \lambda_{k} + \nabla_{x_k}f_k\lambda_{k+1} - \nabla_{x_k}h_k\mu_k,\\
	0=\nabla_{x_k} \mathcal{\bar L} &= \nabla_{x_k} \bar \ell_k - \bar \lambda_{k} + \nabla_{x_k}f_k\bar \lambda_{k+1} - \nabla_{x_k}h_k\bar \mu_k,
	\end{align*}
	such that,~\eqref{eq:dual_rotation} solves the equations above since
	\begin{align*}
	\nabla_{x_k} \bar \ell_k = \nabla_{x_k} \ell_k + \nabla_{x_k} \Lambda_{k} - \nabla_{x_k} f_k \nabla_{x_{k+1}} \Lambda_{k+1}.
	\end{align*}
	For $k=N$, we denote the Jacobian of the active terminal constraints as $G_\mathrm{f}$ to obtain
	\begin{align*}
	0=\nabla_{x_N} \mathcal{L} &= \nabla_{x_N} V_\mathrm{f}(x_N) - \lambda_{N} + G_\mathrm{f} \nu,\\
	0=\nabla_{x_N} \mathcal{\bar L} &= \nabla_{x_N} \bar V_\mathrm{f}(x_N) - \bar \lambda_{N} + G_\mathrm{f} \bar \nu,
	\end{align*}
	which is also satisfied by~\eqref{eq:dual_rotation}. The other optimality conditions do not depend on $\Lambda$ and, therefore, coincide for the original and rotated problem. 
\end{proof}
The fact that the dual solutions of the original and rotated problem differ is often neglected. However, since the Lagrangian Hessian depends on the dual variables, it is expected that NLP solvers do not take the same steps for the two problems. This fact will be proven in Lemma~\ref{lem:different_iterations}.

\subsection{Implications of Rotation on the SQP Iterates}

In this subsection we prove that, if the Hessian approximation does not depend on the Lagrange multipliers, SQP performs the same primal iterates on the original and rotated problem. Since the steps are then fully independent of the Lagrange multipliers, we only require that the primal initial guess is the same for the two problems for the condition to hold. Hessian approximations which depend on the Lagrange multipliers include (a) exact Hessian (eventually regularized to be positive-definite) and (b) BFGS updates~\cite{Nocedal2006}; but not Gauss-Newton approximations are typically used in tracking NMPC. The approximation that we will propose in Section~\ref{sec:hessian_approximation} does also not depend on the dual variables.

\begin{Lemma}
	\label{lem:iterations_coincide}
	Suppose to solve the original problem~\eqref{eq:mpc} and rotated problem~\eqref{eq:rmpc} using SQP based on a Hessian approximation $L^{(i)}$ which does not depend on the dual variables. Then, if the same primal initial guess is used, the primal steps taken on the two problems coincide.
\end{Lemma}
\begin{proof}
	The proof follows along similar lines as Lemma~\ref{lem:rotation}: we prove that the first primal iterate of the two problems coincides and, since $L^{(i)}$ is independent of the dual variables, also subsequent primal steps must coincide. 
	We observe that the difference in cost between the two subproblems is only due to a difference in the gradient term at each time step $k$. The gradient difference
	\begin{align*}
		\nabla \left (\Lambda(x_k)- \Lambda(f(w_k))\right ) &= \matr{c}{I\\0}\nabla \Lambda_k - \nabla f_k \nabla \Lambda_{k+1}
	\end{align*}
	results in the QP subproblem cost difference
	\begin{align*}
		\nabla \Lambda_k^\top x_k - \nabla \Lambda_{k+1}^\top (\nabla_{x} f_k^\top x_k + \nabla_{u} f_k^\top u_k).
	\end{align*}
	An analogous consideration for the final time step $N$ applies, such that the cost difference is given by a telescopic sum. The only surviving term is $\nabla \Lambda_0^\top x_0$ which is constant. In accordance with Lemma~\ref{lem:rotation} and Equation~\eqref{eq:dual_rotation}, we obtain that the difference in the QP Lagrange multipliers at iterate $i$ is $\bar \lambda_{k}^{\mathrm{QP}_i} = \lambda_{k}^{\mathrm{QP}_i} + \nabla \Lambda_k^{(i)}$.
\end{proof}
If exact Hessian is used, one also needs to provide a consistent guess for the dual variables. 
However, even with a consistent initial guess, only the first QP subproblem is guaranteed to provide the same primal solution for the original and rotated problem. Note that, by Lemma~\ref{lem:rotation} though the iterates differ, the solutions of the two problems coincide.

\begin{Lemma}
	\label{lem:different_iterations}
	Suppose to solve the original problem~\eqref{eq:mpc} and rotated problem~\eqref{eq:rmpc} using SQP with exact Hessian. Then, if the same primal initial guess is used and the Lagrange multipliers are initialized as $\mu_k=\mu_\mathrm{s}=0$, $\nu=0$ and $\boldsymbol{\lambda}_s = (\lambda_\mathrm{s},\ldots,\lambda_\mathrm{s})$ for the original problem and $\boldsymbol{\bar \lambda}_s = 0$ for the rotated problem, the first primal iterate coincides, but the subsequent ones do not.
\end{Lemma}
\begin{proof}
	The Hessian of the Lagrangian of the original and rotated problem are respectively
	\begin{align*}
	\nabla^2_{kk} \mathcal{L} &= \nabla^2 \ell_k + \left \langle \nabla^2 f_k, \lambda_k \right \rangle - \left \langle \nabla^2 h_k, \mu_k \right \rangle,  \\
	\nabla^2_{kk} \mathcal{\bar L} &= \nabla^2 \ell_k  + \left \langle \nabla^2 f_k, \bar \lambda_k \right \rangle - \left \langle \nabla^2 f_k, \nabla \Lambda_k \right \rangle  - \left \langle \nabla^2 h_k, \mu_k \right \rangle \\
	& \hspace{4em} + \nabla^2 \Lambda_k -\nabla f_k\nabla^2 \Lambda_{k+1} \nabla f_k^\top.
	\end{align*}
	By~\cite[Theorem 3]{Faulwasser2018} we have $\bar \lambda_k^{(0)} = \lambda_k^{(0)}+\nabla \Lambda_k^{(0)} $. Then,
	\begin{align*}
	\nabla^2 \mathcal{\bar L}_{kk}^{(0)} = \nabla^2 \mathcal{L}_{kk}^{(0)} 
	&+ \nabla^2 \Lambda_k^{(0)} - \nabla f_k^{(0)}\nabla^2 \Lambda_{k+1}^{(0)} \nabla {f_k^{(0)}}^\top, 
	\end{align*}
	where by the double index $kk$ we denote the Hessian block on the diagonal corresponding to stage $k$. 
	Therefore, the Hessian of the Lagrangian of the original and rotated problem only differ by the terms involving the Hessian of the storage functions, while the terms involving the Hessian of the system dynamics coincide. The KKT conditions of the original and rotated QP subproblems imply
	\begin{align*}
	\nabla_{w_k}\mathcal{\bar L}^{(0)} &= \nabla \bar \ell_k^{(0)} -\matr{c}{\bar \lambda_{k}^{\mathrm{QP}_0}\\0} + \matr{c}{A_k^{(0)} \\ B_k^{(0)}}^\top \hspace{-4pt} \bar \lambda_{k+1}^{\mathrm{QP}_0} - \nabla h_k^{(0)}\mu_k^{\mathrm{QP}_0}\\
	&=\nabla_{w_k} \mathcal{L}^{(0)} -\matr{c}{\Delta \lambda_{k}^{\mathrm{QP}_0}\\0} + \matr{c}{A_k^{(0)} \\ B_k^{(0)}}^\top \hspace{-4pt} \Delta \lambda_{k+1}^{\mathrm{QP}_0} \\
	&\hspace{1.6em} +\matr{c}{\nabla^2 \Lambda_{k}^{(0)} x_k^{\mathrm{QP}_0} \\0} - \matr{c}{A_k^{(0)} \\ B_k^{(0)}}^\top \hspace{-4pt} \nabla^2 \Lambda_{k+1}^{(0)} x_{k+1}^{\mathrm{QP}_0} \\
	&\hspace{1.6em} +\matr{c}{\nabla \Lambda_{k}^{(0)} \\0}- \matr{c}{A_k^{(0)} \\ B_k^{(0)}}^\top \hspace{-4pt} \nabla \Lambda_{k+1}^{(0)}, 
	\end{align*}
	such that $\bar \lambda_k^{\mathrm{QP}_0} = \lambda_k^{\mathrm{QP}_0} + \Delta \lambda_k^{\mathrm{QP}_0}$ solves the KKT conditions of the rotated problem, with 
	\begin{align*}
	\Delta \lambda_{k}^{\mathrm{QP}_0} &= \nabla^2 \Lambda_k^{(0)} x_k^{\mathrm{QP}_0} + \nabla \Lambda_k^{(0)} \\
	&= \nabla \Lambda\left (x_k^{(0)} + x_k^{\mathrm{QP}_0}\right ) + O\left (\left \|x_k^{\mathrm{QP}_0}\right \|^2\right ).
	\end{align*}
	Therefore, apart from the trivial case $x_k^{\mathrm{QP}_0}=0$, at the second SQP step we have
	\begin{align*}
	\nabla^2_{kk} \mathcal{\bar L}^{(1)} = \ & \nabla^2_{kk} \mathcal{L}^{(1)} + \left \langle \nabla^2 f^{(1)}_k, \bar \lambda_k^{(1)} - \nabla \Lambda_k^{(1)} - \lambda_k^{(1)} \right \rangle \\
	& \hspace{5em} + \nabla^2 \Lambda_k^{(1)} - \nabla f_k^{(1)}\nabla^2 \Lambda_{k+1}^{(1)} \nabla {f_k^{(1)}}^\top.
	\end{align*}
	Since
	\begin{align*}
	\bar \lambda_k^{(1)} - \nabla \Lambda_k^{(1)} - \lambda_k^{(1)} = O\left (\left \|x_k^{(1)}-x_k^{(0)}\right \|^2\right ) \neq 0,
	\end{align*}
	by using exact Hessian, the cost functions of the two QP subproblems differ by the term
	\begin{align*}
	\sum_{k=0}^{N-1} \frac{1}{2} {x_k^{\mathrm{QP}_1}}^\top \left \langle \nabla^2 f^{(1)}_k, \bar \lambda_k^{(1)} - \nabla \Lambda_k^{(1)} - \lambda_k^{(1)} \right \rangle x_k^{\mathrm{QP}_1},
	\end{align*}
	such that the primal solutions of the QP subproblems relative to the original and rotated problems do not coincide.
\end{proof}

This lemma is of paramount importance, since it establishes that, while we can use the rotated MPC problem to study the stability of fully converged EMPC, when the exact Hessian is used the same does not apply to partially converged schemes such as, e.g., the RTI scheme, since the iterates of the two problems do not coincide. 
This fact will impact in particular the developments of Section~\ref{sec:rti}.

\section{The EMPC Hessian Approximation}
\label{sec:hessian_approximation}

In this section we propose a Hessian approximation for economic MPC based on 
the Hessian of the LETEMPC formulation~\cite{Zanon2016b}, which coincides with 
the Hessian of the rotated economic MPC at the optimal steady state. At the optimal steady state, the reduced Hessian of the two coincides with that of the economic MPC formulation. 

For the MPC problem~\eqref{eq:mpc} the Hessian evaluated at the optimal steady-state is~\cite{Zanon2016b}:
\begin{align*}
\boldsymbol{H} &= \mathrm{blkdiag}(H,\ldots,H,H_\mathrm{f}), \\
H&:=\nabla_{ww}^2 \mathcal{L}_\mathrm{s}(w_\mathrm{s},\lambda_\mathrm{s},\mu_\mathrm{s}), && H_\mathrm{f} := \nabla^2 V_\mathrm{f}(x_\mathrm{s}).
\end{align*}
Since in general $\boldsymbol{H}$ is not positive-definite, a strategy has been proposed in~\cite{Zanon2016b,DeSchutter2020} to eliminate the directions of negative curvature by solving the convex semidefinite program (SDP):%
\begin{subequations}%
	\label{eq:sdp_pos_def_constraints}%
	\begin{align}%
	\min_{\delta P, F, G, \alpha,\beta} \ \ & \beta + \rho_1 \|F\|  + \rho_2 \|G\| \\
	\mathrm{s.t.} \ \ 
	& \alpha H+\mathcal H(\delta P) + \eta C_{\mathbb{A}_\mathrm{s}}^\top FC_{\mathbb{A}_\mathrm{s}} \succeq  I, \label{eq:sdp_pos_def_constraints:pd} \\
	& \alpha H+\mathcal H(\delta P) + \eta C_{\mathbb{A}_\mathrm{s}}^\top FC_{\mathbb{A}_\mathrm{s}} \preceq \beta I, \label{eq:sdp_pos_def_constraints:maxeig} \\
	& \alpha H_\mathrm{f} - \delta P  - \eta D_{\mathbb{A}^\mathrm{f}_\mathrm{s}}^\top G D_{\mathbb{A}^\mathrm{f}_\mathrm{s}}  \succeq I, \label{eq:sdp_pos_def_constraints_terminal_pd}\\
	& \alpha H_\mathrm{f} - \delta P - \eta D_{\mathbb{A}^\mathrm{f}_\mathrm{s}}^\top G D_{\mathbb{A}^\mathrm{f}_\mathrm{s}} \preceq \beta I, \label{eq:sdp_pos_def_constraints_terminal_maxeig}
	\end{align}
\end{subequations}
with user-defined parameters $\rho_1,\rho_2\geq 0$, $\eta\in\{0,1\}$ and
\begin{align}
	\label{eq:hess_rot}
	\mathcal H(\delta P) := \matr{cc}{ A^\top \delta P A - \delta P & A^\top \delta P B \\ B^\top \delta P A & B^\top \delta P B }, 
\end{align}
with
\begin{subequations}
	\label{eq:quadr_expansion}
	\begin{align}
	& A:=\nabla_{x}f(x_\mathrm{s},u_\mathrm{s})^\top, 
	&& B:=\nabla_{u}f(x_\mathrm{s},u_\mathrm{s})^\top, \\ 
	& C:=\nabla_{(x,u)}h(x_\mathrm{s},u_\mathrm{s})^\top, 
	&&D:=\nabla_{x} g_\mathrm{f}(x_\mathrm{s})^\top,
	\end{align}
\end{subequations}
and $\mathbb{A}_\mathrm{s}$, $\mathbb{A}^\mathrm{f}_\mathrm{s}$ the sets of strictly active constraints at steady state: 
\begin{align*}
	\mathbb{A}_\mathrm{s} &:= \{  i  |  h_i(x_\mathrm{s},u_\mathrm{s})=0, \, \mu_i >0    \}, \ \
	\mathbb{A}^\mathrm{f}_\mathrm{s} := \{  i  |  g_i(x_\mathrm{s})=0, \, \nu_i >0    \}.
\end{align*}
While other variations of the formulation can be derived, in Problem~\eqref{eq:sdp_pos_def_constraints}  it is essential to satisfy~\eqref{eq:sdp_pos_def_constraints:pd} with $\alpha >0$. All other constraints are introduced to compute well-conditioned matrices, thus reducing possible inaccuracies in the solution of the MPC problem. A user-friendly tool implementing~\eqref{eq:sdp_pos_def_constraints} is provided in~\cite{DeSchutter2020}. 

We define the positive-definite Hessian approximation as
\begin{align}
\label{eq:pos_def_hess_approx}
\boldsymbol{M}&:=\mathrm{blkdiag}(M,\ldots,M,M_\mathrm{f}) \succ0,\\
	M&:=H+\mathcal H(\delta P)+ \eta C_{\mathbb{A}_\mathrm{s}}^\top FC_{\mathbb{A}_\mathrm{s}} \succ0, \nonumber \\
	M_\mathrm{f}& :=H_\mathrm{f} -\delta P - \eta D_{\mathbb{A}^\mathrm{f}_\mathrm{s}}^\top G D_{\mathbb{A}^\mathrm{f}_\mathrm{s}}\succ0, \nonumber
\end{align}
where we stress that SDP~\eqref{eq:sdp_pos_def_constraints} is solved offline once and $\boldsymbol{M}$ is fixed.

In the following, we denote as Linear-Quadratic (LQ) system, the linear system~\eqref{eq:quadr_expansion} equipped with the quadratic expansion of the cost evaluated at steady-state. Consistently, we denote LQ MPC problem the corresponding MPC problem.

\subsection{Existence and Properties of the Hessian Approximation}
In the following, we first analyze the (mild) conditions under which the proposed Hessian approximation exists. Then, we prove that the convexification of the Hessian does not alter the reduced Hessian at the optimal steady state, suggesting that good convergence properties will be preserved for initial states close to the optimal steady-state.

\begin{Lemma}
	Assume that the LQ MPC problem is stabilizing for all horizons $N$ and satisfies LICQ. Assume additionally that the terminal cost is selected such that Assumpation~\ref{ass:terminal_stabilizing} holds for the LQ system, such that $H_\mathrm{f}$ solves the Lyapunov equation. 
	Then, Problem~\eqref{eq:sdp_pos_def_constraints} does have a solution with $\alpha>0$.
\end{Lemma}
\begin{proof}
%
	The proof, in case of terminal point constraint, is given in~\cite[Theorem 9]{Zanon2016b}. 
	The case of no active constraints at steady-state with a terminal point constraint is covered in~\cite{Zanon2014d} and~\cite[Theorem 7]{Zanon2016b}. In order to introduce a terminal cost, we observe that Assumption~\ref{ass:terminal_stabilizing} for the LQ system implies
	\begin{align}
		\label{eq:lq_lyap_tc}
		\matr{c}{I \\ -K_\mathrm{f}}^\top \big ( \mathcal{H}\left (H_\mathrm{f}\right ) + H \big )\matr{c}{I \\ -K_\mathrm{f}} \succeq 0,
	\end{align}
	where $K_\mathrm{f}=\nabla \kappa_f(x_\mathrm{s})$. 
	We observe that
	\begin{align*}
		\mathcal H(M_\mathrm{f}) + M = \mathcal H\left (H_\mathrm{f}\right ) - \mathcal H(\delta P) + H + \mathcal H(\delta P),
	\end{align*}
	such that~\eqref{eq:lq_lyap_tc} implies
	\begin{align}
		\label{eq:lq_lyap_tc_rot}
		\matr{c}{I \\ -K_\mathrm{f}}^\top \big ( \mathcal H\left (M_\mathrm{f}\right ) + M \big )\matr{c}{I \\ -K_\mathrm{f}} \succeq 0.
	\end{align}
	which in turn entails $M_\mathrm{f}\succ0$, since $M\succ0$ and~\eqref{eq:lq_lyap_tc_rot} is a Lyapunov equation. Consequently, satisfaction of~\eqref{eq:sdp_pos_def_constraints:pd} implies satisfaction of~\eqref{eq:sdp_pos_def_constraints_terminal_pd} for $\alpha>0$ small enough.
	
	We consider now the case in which active constraints are present. Condition~\eqref{eq:sdp_pos_def_constraints:pd} follows from~\cite[Theorem 9]{Zanon2016b}. 
	In this case 
	\begin{equation*}
		\mathcal H(M_\mathrm{f}) + M + \eta \mathcal{H}\left (D_{\mathbb{A}^\mathrm{f}_\mathrm{s}}^\top GD_{\mathbb{A}^\mathrm{f}_\mathrm{s}}\right ) - \eta C_{\mathbb{A}_\mathrm{s}}^\top FC_{\mathbb{A}_\mathrm{s}} = \mathcal H\left (H_\mathrm{f}\right ) + H.
	\end{equation*}
	We then need to prove that
	\begin{equation*}
	\matr{c}{I \\ -K_\mathrm{f}}^\top \big( \mathcal H(M_\mathrm{f}) + M  \big)  \matr{c}{I \\ -K_\mathrm{f}} \succeq 0,
	\end{equation*}
	i.e., that there exists $G$ such that
	\begin{equation}
		\label{eq:condition_on_G}
		\matr{c}{I \\ -K_\mathrm{f}}^\top \big(  \eta \mathcal{H}\left (D_{\mathbb{A}^\mathrm{f}_\mathrm{s}}^\top GD_{\mathbb{A}^\mathrm{f}_\mathrm{s}}\right ) - \eta C_{\mathbb{A}_\mathrm{s}}^\top FC_{\mathbb{A}_\mathrm{s}} \big)  \matr{c}{I \\ -K_\mathrm{f}} \succeq 0.
	\end{equation}
	Positive invariance of the terminal set implies that~\cite{Kolmanovsky1998} 
	\begin{align}
		\label{eq:tc_active_include_path}
		D_{\mathbb{A}^\mathrm{f}_\mathrm{s}} &= \matr{c}{ C_{\mathbb{A}_\mathrm{s}} \matr{c}{I \\ -K_\mathrm{f}} \\ D_1},
	\end{align}
	i.e., at least all path constraints which are strictly active at steady state are also strictly active constraints in the terminal set, under the terminal feedback control law. 
	Consequently, for all $x\in\mathbb{X}_\mathrm{f}$ we have 
	\begin{align*}
		D_{\mathbb{A}^\mathrm{f}_\mathrm{s}} \matr{cc}{A&B}\matr{c}{I \\ -K_\mathrm{f}} x
		&= \zeta D_{\mathbb{A}^\mathrm{f}_\mathrm{s}} x, 
	\end{align*}
	such that
	\begin{align*}
		\matr{c}{I \\ -K_\mathrm{f}}^\top \mathcal{H}\left (D_{\mathbb{A}^\mathrm{f}_\mathrm{s}}^\top GD_{\mathbb{A}^\mathrm{f}_\mathrm{s}}\right )\matr{c}{I \\ -K_\mathrm{f}} = \zeta D_{\mathbb{A}^\mathrm{f}_\mathrm{s}}^\top GD_{\mathbb{A}^\mathrm{f}_\mathrm{s}}.
	\end{align*}
	We use these facts and select
	\begin{align*}
		G \succeq \matr{cc}{F / \zeta & 0 \\ 0 & 0},
	\end{align*} 
	such that, by~\eqref{eq:tc_active_include_path}, we obtain
	\begin{align}
		\label{eq:tc_pd_condition}
		D_{\mathbb{A}^\mathrm{f}_\mathrm{s}}^\top GD_{\mathbb{A}^\mathrm{f}_\mathrm{s}} \succeq \matr{c}{I \\ -K_\mathrm{f}}^\top  C_{\mathbb{A}_\mathrm{s}}^\top FC_{\mathbb{A}_\mathrm{s}} \matr{c}{I \\ -K_\mathrm{f}},
	\end{align}
	which entails~\eqref{eq:condition_on_G}. Therefore, a solution to Problem~\eqref{eq:sdp_pos_def_constraints} exists. 
%
\end{proof}

\begin{Remark}
	Note that a feasible solution to~\eqref{eq:sdp_pos_def_constraints} is obtained by choosing $\delta P=-\nabla^2 \Lambda(x_\mathrm{s})$. 
	Note that~\eqref{eq:sdp_pos_def_constraints:pd} with $F=0$ requires that the Hessian of the rotated stage cost is positive definite and~\eqref{eq:sdp_pos_def_constraints_terminal_pd} requires that $\nabla^2 \bar V_\mathrm{f}(x_\mathrm{s}) \succ 0$. Moreover,~\eqref{eq:sdp_pos_def_constraints:maxeig}, and~\eqref{eq:sdp_pos_def_constraints_terminal_maxeig} can be satisfied by choosing $\beta$ sufficiently large.
\end{Remark}

\begin{Remark}
	With $\mathbb{X}_\mathrm{f}=\{x_\mathrm{s}\}$, both the terminal cost and~\eqref{eq:sdp_pos_def_constraints_terminal_pd} can be removed. Alternatively, since function $V_\mathrm{f}$ can be chosen arbitrarily, one can choose it such that $\nabla^2 V_\mathrm{f}(x_\mathrm{s}) \succ - \delta P$, for any finite $\delta P$.
\end{Remark}
\begin{Remark}
	Whenever strict dissipativity holds, one can choose $\eta=0$, $F=0$. This parameter and variable, however, have been introduced in~\cite{Zanon2016b} as a remedy to a theoretical gap: while in~\cite{Amrit2011a} sufficiency of strict dissipativity has been proven, necessity has been proven in~\cite{Mueller2015} only under the additional assumption that no constraint is active at the optimal steady state. Therefore, in case of active constraints, MPC might be stabilizing even in case strict dissipativity does not hold. For more details on this topic we refer to~\cite{Zanon2016b}. 
\end{Remark}

When solving an NLP by exact-Hessian SQP or interior-point methods, in order to preserve fast convergence it is desirable to avoid modifying the reduced Hessian unless it has some direction of negative curvature. However, not all regularization strategies provide this guarantee, such that convergence could be slowed unnecessarily. 
We prove next a useful property of the proposed Hessian approximation: at steady-state the Hessian $\boldsymbol{M}$ and the exact Hessian $\boldsymbol{H}$ have the same reduced Hessian. We prove this fact in the following lemma, where we denote as $Z_\mathrm{s}$ the nullspace of the dynamic and strictly active path and terminal constraints at steady state. 
\begin{Theorem}
	\label{thm:red_hess}
	For initial guess $\xi_k=\xi_\mathrm{s}$, $\xi\in \{x,u,\lambda,\mu\}$, the Hessian $\boldsymbol{H}$ of the original MPC problem~\eqref{eq:mpc} and the convexified Hessian $\boldsymbol{M}$ share the same reduced Hessian, i.e.,
	\begin{align*}
		Z_\mathrm{s}^\top \boldsymbol{M}Z_\mathrm{s} =Z_\mathrm{s}^\top \boldsymbol{H}Z_\mathrm{s},
	\end{align*}
	where $Z_\mathrm{s}$ is the nullspace of the Jacobian of the initial, dynamic and active path constraints~\eqref{eq:mpc_ic}-\eqref{eq:mpc_pc} evaluated at steady state.
\end{Theorem}
\begin{proof}
	We consider first the case in which there are no active constraints at steady state. In this case, we prove that $\boldsymbol{M}$ can be obtained as a rotation relying on a quadratic storage function. By construction, see Lemma~\ref{lem:rotation}, rotating with any storage function yields $\bar J(w)=J(w)+\Lambda(x_0)$. This implies that the modification spans the range space of the Jacobian of the dynamic constraints~\eqref{eq:mpc_dyn}. Consider the nullspace $Z_0 = \matr{cccc}{I & 0 & \cdots & 0}^\top$ of the initial constraint $x_0-\hat x_0=0$, then $Z_0^\top \nabla^2_{ww}\Lambda(x_0) Z_0 = 0$. Now take $\Lambda(x):= -x^\top \delta P x$, such that
	\begin{align*}
		\nabla^2 \Lambda(x_\mathrm{s}) =&\ - \delta P,\\
		\nabla^2 \Lambda(f(x_\mathrm{s},u_\mathrm{s})) =&\ -\matr{cc}{ A^\top \delta P A & A^\top \delta P B \\ B^\top \delta P A & B^\top \delta P B } \\
		&\ \hspace{5em}+ \left \langle \nabla^2 f(x_\mathrm{s},u_\mathrm{s}) , \nabla \Lambda(x_\mathrm{s}) \right \rangle,
	\end{align*}
	and $\nabla^2 \left (\Lambda(x_\mathrm{s}) - \Lambda(f(x_\mathrm{s},u_\mathrm{s})) \right ) = \mathcal H(\delta P)$, since
	\begin{align*}
		\nabla \Lambda(x_\mathrm{s}) = - \delta P x_\mathrm{s} = 0.
	\end{align*}
	This concludes the proof in the case of no active path constraint at steady state. In case there is some active constraint, the term $\boldsymbol{F}:=\mathrm{blkdiag}(C_{\mathbb{A}_\mathrm{s}}^\top FC_{\mathbb{A}_\mathrm{s}},\ldots,C_{\mathbb{A}_\mathrm{s}}^\top FC_{\mathbb{A}_\mathrm{s}},0)$ is nonzero. However, this term does not modify the reduced Hessian by construction, since it spans (a subspace of) the range space of the Jacobian of the strictly active constraints.
\end{proof}
\begin{Corollary}[of Theorem~\ref{thm:red_hess}]
	\label{cor:steps}
	If the solver is initialized with the steady-state solution $\xi_k=\xi_\mathrm{s}$, $\xi\in \{x,u,\lambda,\mu\}$, the first SQP iterate generated by using the (regularized) exact Hessian $\left [ \nabla^2_{ww} \mathcal{L} \right ]_+$ generates the same primal step as the first SQP iterate using either the Hessian approximation $\boldsymbol{M}$ or $\left [ \boldsymbol{H} \right ]_+$. The dual step also coincides if the latter is used but can be different otherwise.
\end{Corollary}
\begin{proof}
	The first claim follows from Lemma~\ref{thm:red_hess}. The second claim is proven by noting that $\boldsymbol{H}=\nabla^2_{ww} \mathcal{L}(\boldsymbol{w}_\mathrm{s},\boldsymbol{\lambda}_\mathrm{s},\boldsymbol{\mu}_\mathrm{s},0)$, with $\boldsymbol{\cdot}_\mathrm{s}$ the primal-dual trajectory at steady-state. Since the full Hessian coincides, also the dual step coincides. When Hessian approximation $\boldsymbol{M}$ is used, unless $\left [ \boldsymbol{H} \right ]_+=\boldsymbol{M}$ the dual step will be different.
\end{proof}

\subsection{The Economic Real-Time Iteration Scheme}
\label{sec:rti}

In the following, we consider the Real-Time Iteration (RTI) scheme~\cite{Diehl2002b,Diehl2005b}: a popular scheme for real-time NMPC. 
Alternative approaches for real-time NMPC include the Advanced Step NMPC Controller~\cite{Zavala2009} and the continuation/GMRES approach~\cite{Ohtsuka2004}. These approaches are all based on similar ideas: they rely on some form of path-following; 
fast contraction of Newton's method; sufficient regularity of the MPC problem; and a good initial guess constructed using the solution at the previous time. 

The main RTI stability result~\cite[Theorem 6.3]{Diehl2005b} does not prove asymptotic stability, but rather 
\begin{align*}
\lim_{k\to\infty} x_k = x_\mathrm{s}
\end{align*}
The theory relies on 6 Assumptions. 
Since they are rather technical, we only provide an intuitive explanation: 
(a) $\bar \ell(x,u) \geq m\|x\|^2$ for some $m>0$;
(b) is a controllability assumption which guarantees some form of regularity of the MPC problem and is often used to prove stability for MPC
(c) is a standard assumption needed to prove convergence of Newton's method;
(d),(e),(f) make sure that the shift of the MPC solution at the previous time step is a good initial guess for the MPC problem at the current time step. 

The main difficulty in applying the theory for RTI stability to economic MPC is due to the fact that (a) is violated, since $\ell(x,u) \ngeq \alpha(\|x\|)$. One idea to extend the theoretical framework would be to resort to the rotated MPC problem to prove stability exploiting the following corollary.
\begin{Corollary}[of Theorem 6.3 in~\cite{Diehl2005b}]
	\label{cor:rotated_mpc_rti}
	Suppose that the assumptions of Theorem 6.3 in~\cite{Diehl2005b} hold for the rotated MPC problem~\eqref{eq:rmpc}. Then, RTI stability, i.e., $\lim_{k\to\infty} x_k = x_\mathrm{s}$ holds for rotated MPC.
\end{Corollary}
\begin{proof}
	By construction, $\bar \ell(x,u) \geq \alpha(\|x\|)$. Moreover, at $\hat x_0 = x_\mathrm{s}$ the Hessian of the rotated cost is positive-definite. Therefore, on a compact set the rotated cost satisfies $\bar \ell(x,u) \geq m\|x\|^2$.
\end{proof}
 Unfortunately, as proven in Lemma~\ref{lem:different_iterations}, the primal-dual iterations performed on the original and rotated problems do not coincide in general. For Hessian approximations which are independent of the dual variables, however, Lemma~\ref{lem:iterations_coincide} establishes that the primal-dual iterations performed on the original and rotated problem coincide. 
 \begin{Theorem}
 	\label{thm:rti_stability}
 	Suppose that the assumptions of Theorem 6.3 in~\cite{Diehl2005b} hold for the original MPC problem~\eqref{eq:mpc}, with the exception of $\ell(x,u) \geq m\|x\|^2$, which is replaced by strict dissipativity with $\rho(\|x\|)\geq m\|x\|^2$. Then, RTI stability holds, i.e., $\lim_{k\to\infty} x_k = x_\mathrm{s}$, provided that the Hessian approximation is independent of the dual variables.
 \end{Theorem}
\begin{proof}
	The proof follows from Lemma~\ref{lem:iterations_coincide} and Corollary~\ref{cor:rotated_mpc_rti}.
\end{proof}


\begin{Remark}
	\label{rem:eHrti}
	The surprising fact is that this result holds for the GN Hessian approximation, but not for exact Hessian, nor for BFGS Hessian approximations. For such cases, a more refined analysis is required, which is beyond the scope of this paper. 
	However, the RTI stability analysis is to be understood as a theoretical justification supporting the use of RTI in practice. The most important concept can be summarized as: \emph{the fast contraction rate of Newton's method is used to reject the perturbations acting on the closed-loop system}. In this view, though a proof for exact Hessian would require a deeper analysis, one can expect that RTI based on exact Hessian will also be stabilizing.
	Finally, using the identity matrix as Hessian approximation satisfies the theoretical requirements, but might lead to a very small region of attraction and poor performance. 
\end{Remark}

Future work will aim at extending the results to the a more general stability analysis such as, e.g., the one provided in~\cite{Liaomcpherson2020,Zanelli2020}.

\section{Numerical Results}
\label{sec:simulations}


In order to illustrate the theoretical developments, we propose two simple examples. With the first example we verify the results of Lemma~\ref{lem:different_iterations} and Corollary~\ref{cor:steps} and then we compare the different Hessian approximations in closed loop. With the second example, we further illustrate the possible limitations of naive Hessian approximations, such as the identity matrix. 

\subsection{Evaporation Process}
Consider the evaporation process described by states $x = (X_2, \, P_2)$, controls $u=(P_{100},\,F_{200})$ and dynamic equations~\cite{Amrit2013a}:
\begin{align}
M \dot X_2 &= F_1 X_1 - F_2 X_2, &
C \dot P_2 &= F_4-F_5,
\end{align}
where 
\begin{align*}
T_2 &= aP_2 + bX_2 + c, & T_3 &= dP_2 + e, \\
\lambda F_4 &= Q_{100} - F_1C_\mathrm{p}(T_2-T_1), & T_{100} &= fP_{100} + g, \\
Q_{100} &= U_{A_1}(T_{100}-T_2), & U_{A_1} &= h(F_1+F_3), \\
Q_{200} &= \frac{U_{A_2}(T_3-T_{200})}{1+U_{A_2}/(2C_\mathrm{p}F_{200})},& F_{100} &=\frac{Q_{100}}{\lambda_\mathrm{s}} , \\
\lambda F_5 &= Q_{200}, & F_2 &= F_1-F_4.
\end{align*}
All parameters are given in~\cite{Zanon2016b}. 
The economic objective is
\begin{align*}
\ell(x,u) = 10.09(F_2+F_3) + 600 F_{100} + 0.6 F_{200}.
\end{align*}
The system is subject to the following constraints
\begin{align*}
X_2&\geq 25 \,\%, & 40\,\mathrm{kPa}\leq P_2 &\leq 80\,\mathrm{kPa}, \\
P_{100}&\leq 400\,\mathrm{kPa} , & F_{200} & \leq 400 \, \mathrm{kg/min}.
\end{align*}

The optimal steady state is given by
\begin{align}
x_\mathrm{s} = (25, \,49.743), && u_\mathrm{s} = (191.713, \, 215.888).
\end{align}


We use sampling time $t_\mathrm{s}=1 \, \mathrm{min}$ and formulate the NMPC scheme using direct multiple-shooting with a piecewise constant control parametrization, $N=200$ and terminal constraint $x_N=x_\mathrm{s}$. 


\begin{figure}[t]
	\begin{center}
			
		\includegraphics[width=0.9\linewidth]{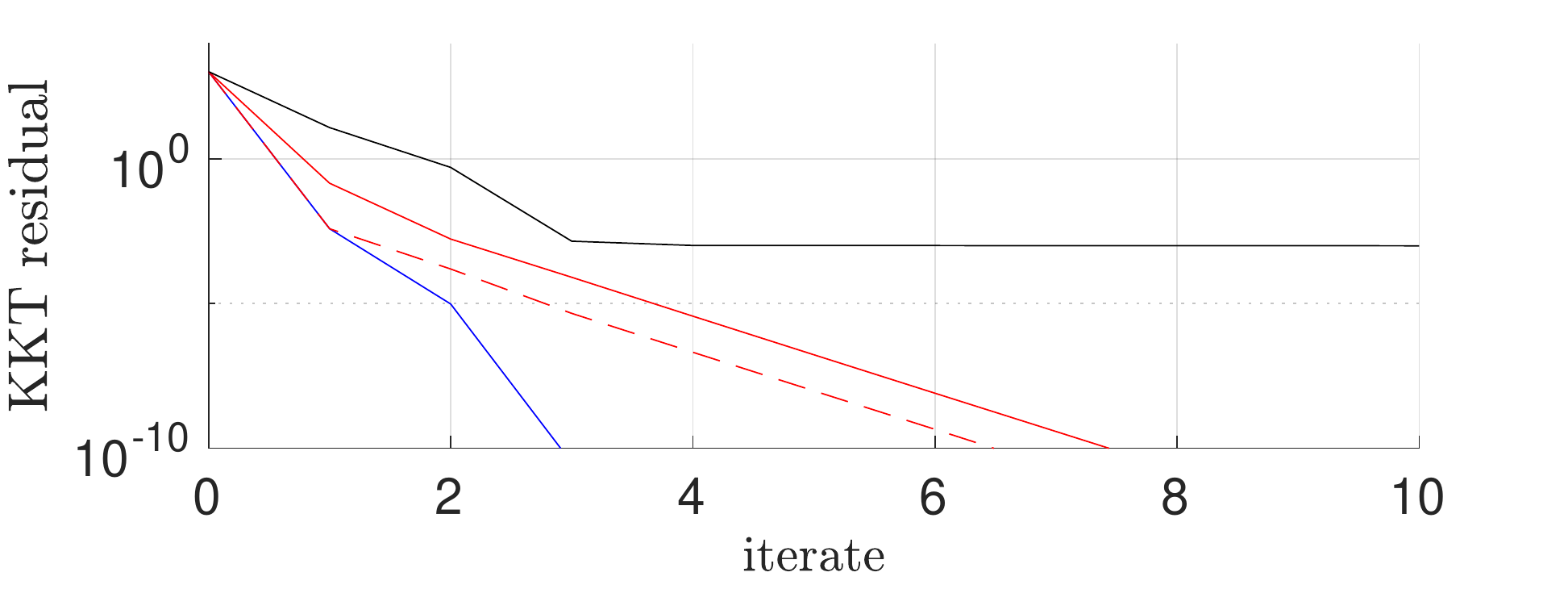}
		\caption{Evaporation process. Convergence for different Hessian approximations: exact (blue), proposed GN (red), proposed GN with dual correction (dashed red), steepest descent (black). 
		}
		\label{fig:evap_convergence}
		
		\includegraphics[width=0.9\linewidth]{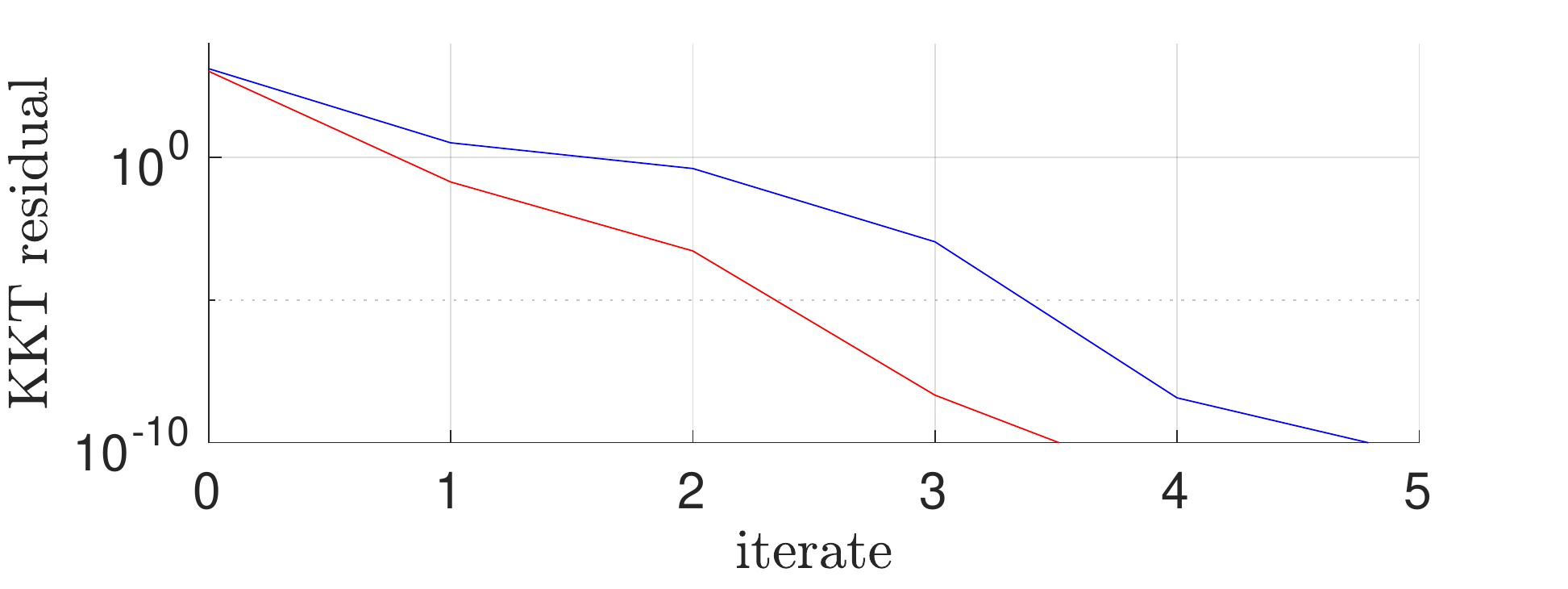}
		\caption{Evaporation process. Convergence for original (blue line) and rotated EMPC (red line) with exact Hessian.}
		\label{fig:evap_rtrack}
		
\end{center}
\end{figure}


Figure~\ref{fig:evap_convergence} displays the convergence of SQP with different Hessian approximations for one instance of the MPC OCP~\eqref{eq:mpc} with $\hat x_0=(35, \,49.743)$. Exact Hessian yields the fastest convergence while steepest descent has very slow convergence. The proposed Hessian approximation has linear convergence with a fast contraction rate, which is typical of Gauss-Newton Hessian approximations see, e.g.,~\cite[Figure~5.1]{Diehl2001}. Since the algorithm is initialized with the steady-state solution, confirming the results of Corollary~\ref{cor:steps}, the primal steps given by exact and GN Hessian coincide, but the dual steps do not. Therefore, the KKT residual after the first iterate is different. Note that~\cite[Algorithm 5]{Verschueren2017} tackles this issue by an ad-hoc computation of the Lagrange multipliers. The convergence with this correction is also displayed in Figure~\ref{fig:evap_convergence}. The primal iterates are unaffected by this correction, which then becomes less useful in an online setting, where one is interested in the primal variables only.

Since the theory also applies to interior-point algorithms, we solved the same problem using Ipopt~\cite{Waechter2009} within CasADi~\cite{Andersson2018}. 
Ipopt converged in $18$ iterations with exact Hessian and $19$ with the GN Hessian; however, the computational times are ${t_\mathrm{c}^\mathrm{Exact}} \approx 2.4{t_\mathrm{c}^\mathrm{GN}}$. 


In order to verify Lemma~\ref{lem:different_iterations}, we constructed an economic MPC formulation artificially by rotating a tracking MPC formulation and solved both EMPC and rotated EMPC with exact Hessian. We used rotated stage cost $\bar \ell(x,u) = 10\,x^\top x + 0.1\, u^\top u$, storage function $\Lambda(x) = 100\,x^\top x$, and $\hat x_0=(35, \,49.743)$. The convergence of SQP is displayed in Figure~\ref{fig:evap_rtrack}, where the rotated formulation converges slightly faster. 

\begin{figure}[t]
	\includegraphics[width=0.9\linewidth]{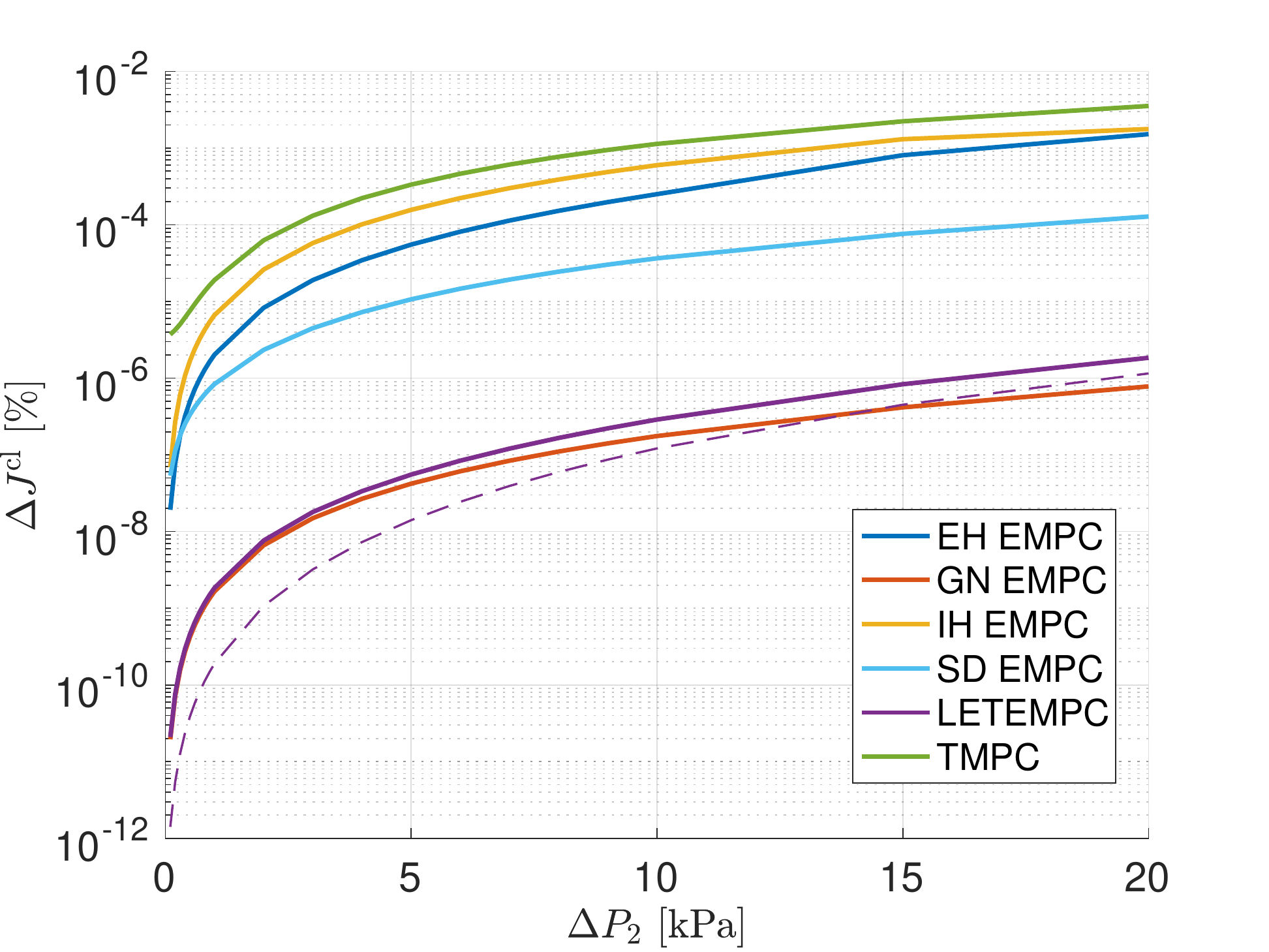}
	\caption{Evaporation process. Closed-loop performance loss for the different EMPC, LETEMPC and TMPC schemes. RTI schemes in continuous line, fully converged in dashed line, baseline: converged EMPC.}
	\label{fig:evap_cl_cost}
\end{figure}

Finally, we consider the cost of closed-loop trajectories obtained with: fully converged EMPC; exact Hessian (EH) RTI EMPC, i.e., $L^{(i)}=\left [ \nabla^2_{ww} \mathcal{\hat L} \right ]_+$; GN RTI EMPC, i.e., $L^{(i)}=\boldsymbol{M}$; indefinite (IH) RTI EMPC, i.e., $L^{(i)}=\boldsymbol{H}$; steepest descent (SD) RTI EMPC, i.e.,  $L^{(i)}=I$; fully converged LETEMPC; GN RTI LETEMPC, i.e., $L^{(i)}=\boldsymbol{M}$; fully converged tracking MPC (TMPC); and GN RTI TMPC. For tracking MPC we $\ell^\mathrm{t}(x,u)=10x^\top x + 0.1 u^\top u$.
We measure performance loss as %
$$\Delta J^\mathrm{cl}:= \frac{J^\mathrm{cl} - J^\mathrm{cl}_\mathrm{EMPC}}{\ell(w_\mathrm{s})N_\mathrm{sim}}100,$$
where $J^\mathrm{cl}_\mathrm{EMPC}$ is the closed-loop cost relative to the fully converged EMPC scheme and $N_\mathrm{sim}$ the simulation duration. 
The simulation results are displayed in Figure~\ref{fig:evap_cl_cost} for initial condition $\hat x_0 = x_\mathrm{s} + (0,\Delta P_2)$, with $\Delta P_2 \in [0,10]$. One can see that GN RTI EMPC is the best among all RTI schemes. Its performance being better than EH RTI EMPC can be explained by the fact that the Hessian regularization procedure was applied on the full Hessian rather than on the reduced one, which is known to slow convergence. The reduced Hessian, however, was not positive definite, such that some form of regularization was necessary. 
Surprisingly, in this example steepest descent did not perform too bad, compared to other schemes. Tracking MPC yields the worst performance, with RTI and fully converged schemes indistinguishable by eye inspection. Finally, the performance of LETEMPC is extremely close to that of GN RTI EMPC.


\subsection{Energy-Optimal Driving}

Consider the following simple electric car model 
\setlength{\extrarowheight}{3pt}
\begin{align*}
	\matr{c}{\dot p_x \\ \dot p_y \\ \dot v \\ \phantom{1} \dot \theta \phantom{1} \\ \dot \delta} &= \matr{c}{ v \cos(\theta)  \\ v \sin(\theta) \\\frac{G_\mathrm{r}}{r}T  - F_\mathrm{b} - F_\mathrm{d}  \\ v\frac{\tan(\delta)}{L} \\ u_\delta } = f_\mathrm{c}\left (\matr{c}{p_x \\  p_y \\ v \\ \phantom{1} \theta \phantom{1} \\ \delta}, \matr{c}{T \\ F_\mathrm{b} \\ u_\delta }\right ),
\end{align*}%
\setlength{\extrarowheight}{0pt}%
where $(p_x,p_y)$ is the position, $v$ the longitudinal velocity, $\theta$ the orientation, $\delta$, $u_\delta$ the steering angle and velocity, $T$ the mechanical torque delivered by the electric motor, $F_\mathrm{b}$ the brake force, and 
$
	F_\mathrm{d} = C_\mathrm{d} v^2 + mgC_\mathrm{r}
$
the drag force due to the aerodynamics and rolling resistance. 
We use parameter values $L=4.8 \ \mathrm{m}$, $m=1700 \ \mathrm{kg}$, $r=0.35 \ \mathrm{m}$, $g=9.81 \ \mathrm{m/s}^2$ $G_\mathrm{r}=7.94$, $C_\mathrm{d}=0.45$ and $C_\mathrm{r}=0.015$. 
We neglect the internal dynamics of electrical motors and assume that the requested torque is delivered instantaneously.
The motor angular velocity is given by
$
	\omega = G_\mathrm{r} \frac{v}{r};
$
and the motor and brakes are subject to the constraints
\begin{align*}
	0 &\leq T \leq \bar T, && T \leq \frac{\bar P}{\omega}, &
	0 &\leq \omega \leq \bar \omega, &
	0 &\leq F_\mathrm{b} \leq \bar F_\mathrm{b},
\end{align*}
where $\bar T= 280 \ \mathrm{Nm}$, $\bar P= 80 \ \mathrm{kW}$, $\bar \omega= 10000 \ \mathrm{rpm}$, $\bar F_\mathrm{b}= 10 \ k\mathrm{N}$. 
We lump all inequality constraints in function $h(x,u)\leq 0$.

The electrical power absorbed by the motor is
\begin{align*}
	P(x,u) = \frac{\omega T}{\eta(\omega,T)}, && \eta(T,\omega)=\frac{\omega T}{\omega T + P_\mathrm{loss}(T,\omega)},
\end{align*}
with efficiency $\eta$ and 
we approximate
\begin{align*}
P_\mathrm{loss}(T,\omega)= 0.0323 \, \omega T + 0.0183 \, \omega^2 + 0.0043 \, T^2,
\end{align*}
where the coefficients have been identified from data~\cite{Murgovski2014}.

In order to minimize fuel consumption while enforcing a prescribed velocity, we adopt the approach proposed in~\cite{Hult2018a} and use the stage cost $\ell(x,u)=\ell^\mathrm{e}(x,u)+\ell^\mathrm{t}(x,u)$ with
\begin{align*}
	\ell^\mathrm{e}(x,u) &= 	P(x,u) + \alpha v \\
	\ell^\mathrm{t}(x,u) &=  b_0 (p_y-p_y^\mathrm{r})^2+  b_1 \theta^2+ b_2 \delta^2 + b_3 u_\delta^2,
\end{align*}
two cost components aiming at minimizing respectively the energy consumption for the prescribed velocity and the lateral deviation from a prescribed reference. For a strategy to choose $\alpha$ so as to enforce that a desired velocity $v^\mathrm{r}$ is attained at the equilibrium we refer to~\cite{Hult2018a}. We choose weights $b_i=1$, $i=0,1,2,3$. 
While many different choices are possible for the cost function, a thorough discussion on the most appropriate choice is beyond the scope of this paper.

We formulate the MPC OCP~\eqref{eq:mpc} in the multiple shooting framework, using a sampling time $t_\mathrm{s}=0.1 \ \mathrm{s}$, a prediction horizon $N=100$ sampling instants discretize the dynamics $f_\mathrm{c}$ using one step of an explicit Runge-Kutta integrator of order $4$ with $5$ steps per control interval to obtain the state transition function $x_{k+1}=f(x_k,u_k)$. For the terminal cost we use the quadratic cost-to-go associated with the LQR formulated at steady state. 

We simulate the system in closed-loop using $v^\mathrm{r}=50 \ \mathrm{km/h}$, i.e., $\alpha=0.055$, and a step reference being $p_y^\mathrm{r}(t) = 0 \ \mathrm{m}$, for $t< 8 \ \mathrm{s}$ and $p_y^\mathrm{r}(t) = \Delta p_y^\mathrm{r}$ for $t\geq 8 \ \mathrm{s}$, with $\Delta p_y^\mathrm{r} \in [0,3] \ \mathrm{m}$. All other references are set to $0$. We introduce an obstacle enforcing the additional constraint $p_x(t)\leq 80 \ \mathrm{m}$, for $t\leq 6 \ \mathrm{s}$.

\begin{figure}
	\includegraphics[width=0.9\linewidth]{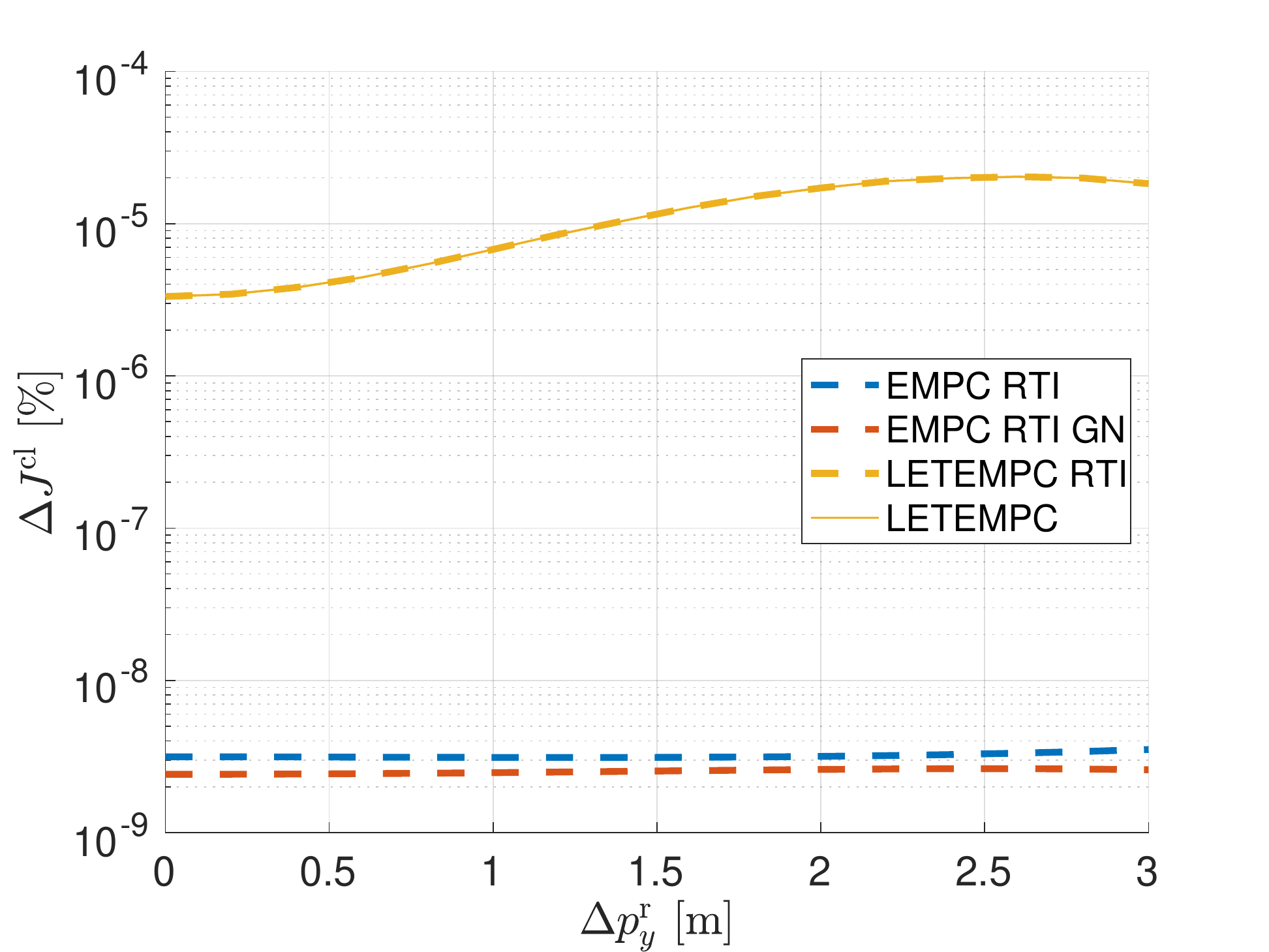}
	\caption{Closed-loop performance for minimum energy driving.}
	\label{fig:veh_cl_cost}
\end{figure}

In Figure~\ref{fig:veh_cl_cost}, we compare various MPC formulations.
As expected EMPC formulations perform better than LETEMPC, though the performance loss is less than $2\cdot 10^{-4} \ \%$. Moreover, the RTI formulation relying on the proposed Hessian approximation performs as well as the one based on exact Hessian. Given its small entity, the slightly better performance is possibly due to numerical inaccuracies. Finally, RTI EMPC based on steepest descent did not stabilize the system, as the iterates diverged.

	We reproduced the simulations with acados~\cite{Verschueren2018}: 
	the sensitivity computation was $2.2$ times faster ($0.9 \, \mathrm{ms}$ vs $2.0 \, \mathrm{ms}$), while the overall RTI step was $1.6$ times faster ($2.9 \, \mathrm{ms}$ vs $4.5 \, \mathrm{ms}$). Note, however, that these numbers depend on the system size, the expression tree complexity, the prediction horizon and the problem formulation. Indeed, in~\cite{DeSchutter2020} the computation times of exact Hessian RTI are reported as $16$ times longer than Gauss-Newton RTI ($213 \ \mathrm{ms}$ vs $13 \ \mathrm{ms}$) for a tethered aircraft.

\section{Conclusions}
\label{sec:conclusions}

In this paper we have investigated efficient algorithms tailored to real-time economic MPC. While the original and rotated EMPC formulations are interchangeable if solved to full convergence, we have proven in theory and verified in practice that they differ if only a limited amount of iterations is performed. 
In order to reduce the computational burden of EMPC, we have proposed a Gauss-Newton-like Hessian approximation which yields fast convergence while only requiring the computation of first-order sensitivities. We have proven that the GN Hessian approximation exists, provided that EMPC is locally stabilizing and we have provided a practical approach to compute it. Simulation results on two examples have demonstrated the effectiveness of the GN Hessian used in RTI EMPC. 

Future research will consider real-time implementations and further investigate the use of exact Hessian with ad-hoc real-time regularization procedures inspired by~\cite{Verschueren2017}.

\bibliographystyle{IEEEtran}
\bibliography{my_bib}

\end{document}